\documentclass[11pt]{article}
\usepackage[margin=1.08in]{geometry}
\usepackage{amsmath,amssymb,amsthm,mathtools}
\usepackage{microtype}
\usepackage{xcolor}
\usepackage{enumitem}
\usepackage[colorlinks=true,linkcolor=blue!55!black,citecolor=blue!55!black,urlcolor=blue!55!black]{hyperref}
\usepackage[T1]{fontenc}
\usepackage{lmodern}

\newtheorem{theorem}{Theorem}[section]
\newtheorem{proposition}[theorem]{Proposition}
\newtheorem{lemma}[theorem]{Lemma}
\newtheorem{corollary}[theorem]{Corollary}

\newcommand{\Pp}{\mathbb P}

\title{The Sharp R\'enyi and Tsallis Threshold \\ in the Shepp--Olkin Concavity Problem}
\author{Haoran Wang\\\small Independent Researcher\\
\small\texttt{whr.hrwang@gmail.com}}
\date{}

\begin{document}
\maketitle

\begin{abstract}
Let $B_1,\ldots,B_n$ be independent Bernoulli random variables with parameters $p_1,\ldots,p_n$, and let $S=\sum_i B_i$. Hillion and Johnson proved that the Shannon entropy of $S$ is jointly concave in the parameter vector and conjectured analogous critical ranges for R\'enyi and Tsallis entropies. We determine the exact range. For every real $0<q<1$, the power sum $\sum_k \Pp(S=k)^q$ is jointly concave, and it is strictly concave on the open parameter cube. Hence both the R\'enyi and Tsallis entropies of order $q$ are jointly concave. At $q=1$ this is the known Shannon theorem. For every real $q>1$, a two-variable transverse interpolation gives strict local convexity for both entropies. Thus the universal concavity range is exactly $0<q\le 1$ for both families. The proof below order one uses the transport inequality of Hillion and Johnson together with a nonlinear telescoping correction. A Riccati identity reduces the remaining local estimate to a two-by-two determinant whose zeros have a one-sided crossing rule.
\end{abstract}

\section{Introduction}

Entropy under convolution measures how uncertainty changes when independent random inputs are combined. For sums of independent Bernoulli variables, the Shepp--Olkin problem asks for a stronger geometric property: whether this uncertainty is jointly concave in the individual success probabilities.

Let
\[
B_i\sim \operatorname{Bernoulli}(p_i),\qquad i=1,\ldots,n,
\]
be independent, and write
\[
S=\sum_{i=1}^n B_i,\qquad f_k=\Pp(S=k).
\]
The classical Shepp--Olkin conjecture asserts that the Shannon entropy of $S$ is jointly concave in $(p_1,\ldots,p_n)$. Thus concavity must hold along every affine direction in the parameter cube, including directions in which some Bernoulli parameters increase while others decrease. Hillion and Johnson proved the conjecture in full \cite{HJ17}, building on their discrete transport framework and their earlier proof for monotone interpolations \cite{HJ16}.

R\'enyi and Tsallis entropies give two standard one-parameter extensions of Shannon entropy \cite{Ren61,Tsa88}. For $q>0$, $q\ne1$, set
\[
H_q(S)=\frac{1}{1-q}\log\sum_k f_k^q,
\qquad
T_q(S)=\frac{1-\sum_k f_k^q}{q-1}.
\]
Both converge to Shannon entropy as $q\to1$. Hillion and Johnson proposed a generalized Shepp--Olkin conjecture with critical orders. They conjectured the R\'enyi threshold to be $2$ and the Tsallis threshold to be the root $q\approx3.65986$ of $2^q-4q+2=0$ \cite{HJ17}.

Both thresholds are in fact equal to $1$. Below the Shannon order, the common power sum $\sum_k f_k^q$ is jointly concave for every $0<q<1$. Above it, joint concavity already fails for two Bernoulli variables along a direction in which their parameters move oppositely. Thus $q=1$ is the sharp boundary between joint concavity and transverse instability.

\subsection{Prior work}

Shepp and Olkin formulated the joint entropy-concavity problem for sums of independent Bernoulli variables in 1981 \cite{SO81}. Hillion and Johnson later developed a discrete transport method and proved joint concavity along monotone interpolations \cite{HJ16}. They then established full joint concavity for Shannon entropy \cite{HJ17}.

The 2017 paper \cite{HJ17} also formulated the R\'enyi and Tsallis extension considered here. Its proof provides two ingredients that we use: a transport representation for the first two derivatives of a Poisson-binomial mass function, and a sharp local bound on the second transport term. The same paper observed that adding a telescoping discrete derivative might help with generalized entropies; see Remark 4.5 of \cite{HJ17}. Our argument develops this suggestion in a nonlinear form. The new ingredients are an explicit correction valid for every $0<q<1$, together with a Riccati identity and a zero-crossing argument that establish the required local positivity.

Hillion and Johnson subsequently proved the Shepp--Olkin monotonicity conjecture for Shannon entropy \cite{HJ19}. They also studied R\'enyi and Tsallis analogues of this monotonicity question. This provides a useful contrast with joint concavity. Their monotonicity problem concerns parameter changes toward fairer coins, whereas joint concavity must control every affine direction, including mixed-sign ones. This distinction is relevant here: our counterexample above order one is transverse, with two parameters moving in opposite directions.

There is also a separate literature on R\'enyi entropy of Bernoulli sums. Madiman, Melbourne and Roberto obtained sharp Fourier-analytic entropy--variance comparisons and applications to entropy power and anti-concentration \cite{MMR23}. Those results concern different functionals and do not address joint concavity in the full Bernoulli parameter vector.

\subsection{Our results}

The main result is the following sharp classification.

\begin{theorem}\label{thm:main}
For every $n\ge1$, let $S=\sum_{i=1}^n B_i$ be a sum of independent Bernoulli variables. The following statements hold for real $q>0$.
\begin{enumerate}[label=\textup{(\roman*)}]
\item If $0<q<1$, then
\[
(p_1,\ldots,p_n)\longmapsto \sum_k \Pp(S=k)^q
\]
is jointly concave on $[0,1]^n$ and strictly concave on $(0,1)^n$. Consequently, both $H_q(S)$ and $T_q(S)$ are jointly concave, and they are strictly concave on $(0,1)^n$.
\item At $q=1$, both families reduce to Shannon entropy, which is jointly concave by Hillion and Johnson \cite{HJ17}.
\item If $q>1$, then joint concavity fails already for the sum of two Bernoulli variables. More precisely, there is an interior affine path on which both $H_q$ and $T_q$ have strictly positive second derivative at its midpoint.
\end{enumerate}
Hence the universal joint-concavity range for both R\'enyi and Tsallis entropy is exactly
\[
\boxed{\ 0<q\le1.\ }
\]
\end{theorem}

The negative part has a short exact witness in two variables, while the positive part is the main argument of the paper. In particular, the two-variable obstruction embeds into every higher-dimensional parameter cube, so joint concavity on $[0,1]^n$ fails for every $n\ge2$ when $q>1$.

\begin{corollary}
The critical orders in the generalized Shepp--Olkin problem of \cite{HJ17} are
\[
q_R^*=q_T^*=1.
\]
\end{corollary}

\subsection{Technique overview}

Fix an affine parameter path $p_i(t)$ and denote the corresponding mass function by $f_k(t)$. Following \cite{HJ16,HJ17}, the first- and second-order changes of this law can be encoded by transport coefficients $g_k$ and $h_k$ satisfying
\[
f_k'=g_{k-1}-g_k,
\qquad
f_k''=h_k-2h_{k-1}+h_{k-2}.
\]
Thus the parameter derivatives become discrete divergences on the probability axis. Hillion and Johnson proved the local transport inequality
\[
h_k(f_{k+1}^2-f_kf_{k+2})
\le 2f_{k+1}g_kg_{k+1}-f_{k+2}g_k^2-f_kg_{k+1}^2.
\]
For $q=1-s$ with $0<s<1$, the second derivative of the power sum is a sum of local quadratic terms in $(g_k,g_{k+1})$ and $h_k$. A direct termwise estimate is too strong. We add a telescoping correction
\[
\eta_s(a,b)=\frac{1-s}{s}\,
\frac{(a^{-s}-b^{-s})^2}{b^{1-s}-a^{1-s}}.
\]
After the correction and the transport bound, each local term is bounded below by a two-by-two quadratic form. The determinant of this form is the only remaining issue.

The determinant is governed by a one-variable function $A_s$, chosen so that the endpoint coefficients and the local matrix can be handled in the same notation. This function satisfies an exact Riccati identity together with a reciprocal relation under $x\mapsto x^{-1}$. The reciprocal relation controls the degenerate boundary, while the Riccati identity fixes the orientation of any hypothetical interior zero. Since the determinant is positive just inside the boundary and every interior zero could only be crossed in the opposite direction, no such zero exists. This proves local positive definiteness for the whole interval $0<s<1$.

The argument above order one is different. A transverse two-coin interpolation keeps the first derivative of the probability vector equal to zero at its midpoint. Its second derivative has the wrong sign for every $q>1$.

\section{Preliminaries}

We work first at an interior point $(p_1,\ldots,p_n)\in(0,1)^n$ and along an affine path
\[
p_i(t)=p_i+t b_i.
\]
Boundary points follow by continuity.

For each $i$, let $f_k^{(i)}(t)$ denote the probability that the sum with $B_i$ removed equals $k$. Likewise, for $i\ne j$, let $f_k^{(i,j)}(t)$ be the mass function after removing both $B_i$ and $B_j$. Define
\[
g_k(t)=\sum_i b_i f_k^{(i)}(t),
\qquad
h_k(t)=\sum_{i\ne j} b_i b_j f_k^{(i,j)}(t),
\]
where the second sum is over ordered pairs $i\ne j$. We suppress the dependence on $t$ from now on and let all these sequences vanish outside their natural support. Direct differentiation then gives
\begin{equation}\label{eq:transport-derivatives}
f_k'=g_{k-1}-g_k,
\qquad
f_k''=h_k-2h_{k-1}+h_{k-2}.
\end{equation}

Along a general affine direction, the second-order coefficient $h_k$ has no fixed sign. The key input from Hillion and Johnson is that strict log-concavity of the Poisson-binomial law allows $h_k$ to be controlled by the neighboring first-order coefficients $g_k$ and $g_{k+1}$.

\begin{proposition}[Hillion--Johnson \cite{HJ17}]\label{prop:HJ}
For $0\le k\le n-2$,
\begin{equation}\label{eq:HJ}
h_k(f_{k+1}^2-f_kf_{k+2})
\le
2f_{k+1}g_kg_{k+1}-f_{k+2}g_k^2-f_kg_{k+1}^2.
\end{equation}
\end{proposition}

At an interior parameter point, a Poisson-binomial mass function is strictly log-concave on its support. In particular,
\begin{equation}\label{eq:slc}
f_{k+1}^2>f_kf_{k+2},\qquad 0\le k\le n-2.
\end{equation}
This follows from Newton's inequalities; see also \cite{HJ17}.

\section{Proof of the sharp range}

\subsection{All orders above one fail}

Fix $q>1$. Let
\[
p=\frac{1}{1+2^{q/(q-1)}}
\]
and consider
\[
p_1(t)=p+t,
\qquad
p_2(t)=p-t.
\]
For small $|t|$ this is an interior path. If $F_q(t)=\sum_{k=0}^2 f_k(t)^q$, then
\[
f_0(t)=(1-p)^2-t^2,
\quad
f_1(t)=2p(1-p)+2t^2,
\quad
f_2(t)=p^2-t^2.
\]
Hence $F_q'(0)=0$ and
\begin{align*}
F_q''(0)
&=-2q\Big((1-p)^{2q-2}+p^{2q-2}
 -2(2p(1-p))^{q-1}\Big)\\
&=-2q\,p^{2q-2}<0.
\end{align*}
The last identity is exactly the choice of $p$. Therefore
\[
T_q''(0)=-\frac{F_q''(0)}{q-1}>0,
\]
and, since $F_q'(0)=0$,
\[
H_q''(0)=\frac{F_q''(0)}{(1-q)F_q(0)}>0.
\]
Thus both entropy families fail joint concavity for every real $q>1$.

\subsection{A nonlinear telescoping correction below one}

Now fix $0<q<1$ and write
\[
s=1-q\in(0,1).
\]
The case $n=1$ is immediate from the strict concavity of $p^{1-s}+(1-p)^{1-s}$, so assume $n\ge2$. Let
\[
\Phi(t)=\sum_{k=0}^n f_k(t)^{1-s}.
\]
We prove $\Phi''(t)\le0$ for every affine path.

Using \eqref{eq:transport-derivatives} and summing twice by parts gives
\begin{align}
-\frac{\Phi''}{(1-s)s}
={}&\sum_{k=0}^n f_k^{-1-s}(g_{k-1}-g_k)^2 \notag\\
&-\frac1s\sum_{k=0}^{n-2}h_k
\left(f_k^{-s}-2f_{k+1}^{-s}+f_{k+2}^{-s}\right).
\label{eq:curvature}
\end{align}

For $a,b>0$ with $a\ne b$, define
\begin{equation}\label{eq:eta}
\eta_s(a,b)=\frac{1-s}{s}\,
\frac{(a^{-s}-b^{-s})^2}{b^{1-s}-a^{1-s}}.
\end{equation}
The expression has a continuous extension to the diagonal, and we set $\eta_s(a,a)=0$. It is antisymmetric and homogeneous of degree $-1-s$.

For a fixed $k$, abbreviate
\[
a=f_k,\quad b=f_{k+1},\quad c=f_{k+2},
\quad g=g_k,\quad \widetilde g=g_{k+1}.
\]
Define
\begin{align*}
u_k={}&a^{-1-s}g^2-2b^{-1-s}g\widetilde g+c^{-1-s}\widetilde g^2\\
&-\frac{h_k}{s}(a^{-s}-2b^{-s}+c^{-s}),
\end{align*}
and add the correction
\[
\widehat u_k=u_k-\eta_s(a,b)g^2+\eta_s(b,c)\widetilde g^2.
\]
The correction telescopes. A direct expansion gives
\begin{align}
-\frac{\Phi''}{(1-s)s}
={}&\sum_{k=0}^{n-2}\widehat u_k
+\big(f_1^{-1-s}+\eta_s(f_0,f_1)\big)g_0^2 \notag\\
&+\big(f_{n-1}^{-1-s}-\eta_s(f_{n-1},f_n)\big)g_{n-1}^2.
\label{eq:telescope}
\end{align}

The endpoint coefficients and the local matrix are governed by the same one-variable quantity, which motivates the following definition.

\begin{lemma}\label{lem:As}
For $x>0$, let
\[
L_s(x)=x^{-1-s}-\eta_s(x,1),
\qquad
A_s(x)=xL_s(x).
\]
Then $A_s(x)\ge1$, with equality only at $x=1$. Moreover
\begin{equation}\label{eq:riccati}
xA_s'(x)
=A_s(x)-x^{-s}
-sx^{-s}\frac{(A_s(x)-1)^2}{(1-x^{-s})^2},
\end{equation}
with the right-hand side interpreted by continuity at $x=1$, and
\begin{equation}\label{eq:reciprocal}
A_s(x)+A_s(x^{-1})
=2+\frac{x^{-s}+x^s-2}{s}.
\end{equation}
\end{lemma}

\begin{proof}
From \eqref{eq:eta},
\begin{equation}\label{eq:Aformula}
A_s(x)
=1+\frac{x^{-s}-1}{s}
\left(1-\frac{(1-s)(x-1)}{x^{1-s}-1}\right).
\end{equation}
By the mean value theorem,
\[
\frac{(1-s)(x-1)}{x^{1-s}-1}=\xi^s
\]
for some $\xi$ between $1$ and $x$. The two factors after the leading $1$ in \eqref{eq:Aformula} have the same sign, proving $A_s\ge1$ and the equality statement. Differentiating \eqref{eq:Aformula} and simplifying gives \eqref{eq:riccati}. Replacing $x$ by $x^{-1}$ in \eqref{eq:Aformula} and adding gives \eqref{eq:reciprocal}.
\end{proof}

By homogeneity, the boundary coefficients in \eqref{eq:telescope} are positive multiples of $L_s(x)$, hence are strictly positive by Lemma \ref{lem:As}. It remains to control each $\widehat u_k$.

Since $b^2>ac$, the arithmetic--geometric mean inequality gives
\[
a^{-s}-2b^{-s}+c^{-s}>0.
\]
In the curvature formula, the $h_k$ term is multiplied by the negative coefficient $-\bigl(a^{-s}-2b^{-s}+c^{-s}\bigr)/s$. Multiplying \eqref{eq:HJ} by this coefficient therefore reverses the inequality and yields
\begin{equation}\label{eq:local-matrix}
\widehat u_k\ge
\begin{pmatrix}g&\widetilde g\end{pmatrix}
M
\begin{pmatrix}g\\\widetilde g\end{pmatrix},
\end{equation}
where
\[
M=\begin{pmatrix}
a^{-1-s}-\eta_s(a,b)+cK & -b^{-1-s}-bK\\
-b^{-1-s}-bK & c^{-1-s}+\eta_s(b,c)+aK
\end{pmatrix}
\]
and
\[
K=\frac{a^{-s}-2b^{-s}+c^{-s}}{s(b^2-ac)}>0.
\]
After dividing by the positive factor $b^{-1-s}$ and replacing $a/b,c/b$ by $a,c$, it remains to prove the following local statement.

\begin{lemma}\label{lem:matrix}
For $0<s<1$ and $a,c>0$ with $ac<1$, set
\[
K=\frac{a^{-s}+c^{-s}-2}{s(1-ac)}.
\]
Then
\begin{equation}\label{eq:Msc}
M_s(a,c)=
\begin{pmatrix}
L_s(a)+cK&-1-K\\
-1-K&L_s(c)+aK
\end{pmatrix}
\end{equation}
is positive definite.
\end{lemma}

\begin{proof}
The diagonal entries are positive because $L_s>0$ and $K>0$. Thus positive definiteness is equivalent to positivity of the determinant.

Write $z=ac$ and
\begin{equation}\label{eq:D}
D(a,c)=\det M_s(a,c)
=\frac{(A_s(a)+zK)(A_s(c)+zK)}{z}-(1+K)^2.
\end{equation}
The determinant tends to zero as $z\uparrow1$. To determine its sign on the interior side of this boundary, fix $r>0$, $r\ne1$, and move inward along $(a,c)=(re^{-\delta},r^{-1})$ with $\delta>0$. Using \eqref{eq:riccati} and \eqref{eq:reciprocal}, a first-order expansion as $\delta\downarrow0$ gives
\begin{equation}\label{eq:boundary-expansion}
D(re^{-\delta},r^{-1})
=\frac{\delta}{2s}
\left[(1+s)-(1-s)\frac{r-r^{-s}}{r^{1-s}-1}\right]^2
+O_{r,s}(\delta^2).
\end{equation}
The square is nonzero. Indeed, if $u=\log r$, then
\[
\frac{r-r^{-s}}{r^{1-s}-1}
=\frac{\sinh((1+s)u/2)}{\sinh((1-s)u/2)},
\]
and $\sinh t/t$ is strictly increasing on $(0,\infty)$. Hence the coefficient of $\delta$ in \eqref{eq:boundary-expansion} is strictly positive. For fixed $r\ne1$ and $s\in(0,1)$, the remainder is bounded in absolute value by $C_{r,s}\delta^2$ for all sufficiently small $\delta$. The positive linear term therefore dominates, and $D>0$ just inside every boundary point $ac=1$ with $a\ne c$.

We next show that an interior zero can only be crossed in the opposite direction. Suppose $D(a,c)=0$. Put
\[
x=a^{-s},\qquad y=c^{-s},\qquad d=s(1-z),
\]
and
\[
u=\frac{A_s(a)+zK}{1+K}>0.
\]
The equation $D=0$ implies
\[
\frac{A_s(c)+zK}{1+K}=\frac{z}{u}.
\]
If $x\ne1$, differentiation of \eqref{eq:D} with respect to $\log a$, followed by \eqref{eq:riccati}, gives
\begin{align}
\frac{u(1-z)}{1+K}\,\frac{\partial D}{\partial\log a}
={}&-\frac{xE}{(1-x)^2}
\left[(u-1)+\frac{dK}{E}(u-z)\right]^2 \notag\\
&+\frac{K(1+d-xy)}{E}(u-z)^2,
\label{eq:crossing}
\end{align}
where $E=(1-x)^2+d$. Since $xy=z^{-s}$ and $0<z<1$, strict convexity of $t\mapsto t^{-s}$ at $1$ gives
\[
xy=z^{-s}>1+s(1-z)=1+d.
\]
Thus both terms on the right of \eqref{eq:crossing} are nonpositive. They cannot vanish simultaneously: if $u=z$, then $u-1=z-1\ne0$. Hence
\[
\frac{\partial D}{\partial\log a}<0.
\]
If $x=1$, the removable case can be evaluated directly. Then
\[
K=\frac{z^{-s}-1}{s(1-z)}>1
\]
and at a zero of $D$,
\[
\frac{u}{1+K}\frac{\partial D}{\partial\log a}
=\frac{(1-z)K(1-K)}{(1+K)^2}<0.
\]

Fix $c\ne1$ and vary $a\in(0,1/c)$. By \eqref{eq:boundary-expansion}, the determinant is positive near the right endpoint. If it were nonpositive somewhere, continuity would force a zero at which $D$ passes from nonpositive values to positive values as $a$ increases. The derivative at such a crossing must be nonnegative, contradicting the strict crossing rule above. Hence $D>0$ whenever $c\ne1$.

For $c=1$, continuity from nearby values of $c$ first gives $D(a,1)\ge0$. Equality at an interior point would be impossible: the strict derivative rule would make $D$ negative immediately on one side of that point. Therefore $D(a,1)>0$ as well, and $M_s(a,c)$ is positive definite throughout $ac<1$.
\end{proof}

Lemma \ref{lem:matrix} and \eqref{eq:local-matrix} show $\widehat u_k\ge0$ for every $k$. Together with the positive boundary terms in \eqref{eq:telescope}, this gives $\Phi''\le0$. Hence $\Phi$ is jointly concave on the interior parameter cube, and continuity extends the result to $[0,1]^n$.

The concavity is strict in the interior. Let $(b_i)$ be a nonzero direction. If the sequence $g$ is not identically zero, then either an endpoint coefficient in \eqref{eq:telescope} is multiplied by a nonzero endpoint value, or some local pair $(g_k,g_{k+1})$ is nonzero. In either case, the strict positivity of the endpoint coefficient or the positive definiteness from Lemma \ref{lem:matrix} gives $-\Phi''>0$.

It remains to consider the possible kernel $g\equiv0$. In this case $g_0=f_0\sum_i b_i/(1-p_i)=0$ and
\[
h_0=f_0\left[\left(\sum_i\frac{b_i}{1-p_i}\right)^2-\sum_i\frac{b_i^2}{(1-p_i)^2}\right]<0.
\]
Proposition \ref{prop:HJ} also gives $h_k\le0$ for every $k$. Since $a^{-s}-2b^{-s}+c^{-s}>0$ for each consecutive triple of probabilities, the $k=0$ local term is then strictly positive. Thus $\Phi''<0$ in every nonzero direction at every point of $(0,1)^n$.

Since
\[
T_q=\frac{\Phi-1}{1-q},
\]
Tsallis entropy is jointly concave and strictly concave on the open cube. Also $\log$ is increasing and concave on $(0,\infty)$, so
\[
H_q=\frac{\log\Phi}{1-q}
\]
is jointly concave, and strict concavity of $\Phi$ gives strict concavity on the open cube. This proves Theorem \ref{thm:main} for $0<q<1$.

\section{Discussion}

The generalized Shepp--Olkin problem has a different phase boundary from the one suggested by the first known counterexamples. The obstruction above order one is transverse: it varies two Bernoulli parameters in opposite directions and is already visible for two coins. This explains why one-parameter or monotone tests do not locate the correct threshold.

Below order one, the proof keeps the transport inequality from the Shannon theory but changes the local algebra. The useful correction is not a linear perturbation of the Shannon argument. Its role is to redistribute neighboring quadratic terms until the Hillion--Johnson bound becomes a positive two-by-two form. The function $A_s$ then exposes an exact Riccati structure. The boundary expansion and the one-sided zero-crossing rule are the two parts that close the argument for the whole interval $0<s<1$.

The theorem concerns Bernoulli sums and joint parameter concavity. A natural next question is whether the correction--Riccati mechanism extends to other ultra-log-concave convolution families. Another direction is quantitative stability: one may ask for explicit lower bounds on the negative Hessian in terms of the parameter vector and the interpolation direction.

\section*{Acknowledgements}
The author used OpenAI's GPT-5.6 Sol extensively as a research assistant for exploratory calculations, symbolic verification, literature searches, and drafting. AI-assisted exploration first identified the order $q=1/2$ as a tractable case. Under the author's guidance and structural observations, this was extended to the correction--Riccati argument for the full interval $0<q<1$. The author independently verified the mathematical arguments and references and takes full responsibility for the paper.

\appendix
\section{Algebraic details for the local determinant}\label{app:algebra}

For completeness, we record two calculations used in Lemma \ref{lem:matrix}.

First, starting from \eqref{eq:Aformula}, differentiation gives
\[
xA_s'(x)-A_s(x)+x^{-s}
=-sx^{-s}\frac{(A_s(x)-1)^2}{(1-x^{-s})^2}.
\]
This is \eqref{eq:riccati}. The identity remains valid at $x=1$ by continuity, where $A_s(1)=1$ and $A_s'(1)=0$.

Second, consider the inward curve $(a,c)=(re^{-\delta},r^{-1})$ with fixed $r\ne1$. Along this curve $z=e^{-\delta}$, and the quantity $K$ has a simple pole of order $1/\delta$ as $\delta\downarrow0$. Rewriting
\[
D=\frac{A_s(a)A_s(c)}{z}-1+K\bigl(A_s(a)+A_s(c)-2\bigr)+(z-1)K^2
\]
makes the cancellations visible. The reciprocal identity cancels the singular $1/\delta$ contribution, while the Riccati identity cancels the constant term. The first nonzero coefficient is therefore linear in $\delta$ and reduces to
\[
\frac{1}{2s}
\left[(1+s)-(1-s)\frac{r-r^{-s}}{r^{1-s}-1}\right]^2.
\]
The zero-crossing calculation is obtained by differentiating \eqref{eq:D} at $D=0$. If
\[
x=a^{-s},\quad y=c^{-s},\quad z=ac,\quad d=s(1-z),
\quad E=(1-x)^2+d,
\]
and
\[
u=\frac{A_s(a)+zK}{1+K},
\]
then elimination of $A_s(c)$ by the equation $D=0$ yields exactly \eqref{eq:crossing}. No inequality is used before the last step, where $z^{-s}>1+s(1-z)$ supplies the strict sign.

\section{Boundary parameters}

The proof above was written for $p_i\in(0,1)$. For fixed $q\in(0,1)$, the map
\[
(p_1,\ldots,p_n)\mapsto \sum_k \Pp(S=k)^q
\]
is continuous on the closed cube. Joint concavity on the open cube therefore extends to the closure by approximating any two boundary points by interior points and passing to the limit in the concavity inequality. The same argument applies to R\'enyi and Tsallis entropy because the power sum is positive and both outer functions are continuous on its range.


\begin{thebibliography}{MMR23}

\bibitem[HJ16]{HJ16}
E. Hillion and O. T. Johnson.
Discrete versions of the transport equation and the Shepp--Olkin conjecture.
\emph{Ann. Probab.} 44 (2016), no. 1, 276--306.

\bibitem[HJ17]{HJ17}
E. Hillion and O. T. Johnson.
A proof of the Shepp--Olkin entropy concavity conjecture.
\emph{Bernoulli} 23 (2017), no. 4B, 3638--3649.


\bibitem[HJ19]{HJ19}
E. Hillion and O. T. Johnson.
\newblock A proof of the Shepp--Olkin entropy monotonicity conjecture.
\newblock \emph{Electron. J. Probab.} 24 (2019), paper no. 126, 1--14.


\bibitem[MMR23]{MMR23}
M. Madiman, J. Melbourne and C. Roberto.
Bernoulli sums and R\'enyi entropy inequalities.
\emph{Bernoulli} 29 (2023), no. 2, 1578--1599.

\bibitem[Ren61]{Ren61}
A. R\'enyi.
On measures of entropy and information.
In \emph{Proc. Fourth Berkeley Symp. Math. Statist. Probab.}, Vol. I, 547--561, 1961.

\bibitem[SO81]{SO81}
L. A. Shepp and I. Olkin.
Entropy of the sum of independent Bernoulli random variables and of the multinomial distribution.
In \emph{Contributions to Probability}, 201--206, Academic Press, 1981.

\bibitem[Tsa88]{Tsa88}
C. Tsallis.
Possible generalization of Boltzmann--Gibbs statistics.
\emph{J. Stat. Phys.} 52 (1988), 479--487.

\end{thebibliography}
\end{document}